\documentclass[letterpaper]{article} % DO NOT CHANGE THIS
\usepackage{aaai23}  % DO NOT CHANGE THIS
\usepackage{times}  % DO NOT CHANGE THIS
\usepackage{helvet}  % DO NOT CHANGE THIS
\usepackage{courier}  % DO NOT CHANGE THIS
\usepackage[hyphens]{url}  % DO NOT CHANGE THIS
\usepackage{graphicx} % DO NOT CHANGE THIS
\usepackage{natbib}  % DO NOT CHANGE THIS AND DO NOT ADD ANY OPTIONS TO IT
\usepackage{caption} % DO NOT CHANGE THIS AND DO NOT ADD ANY OPTIONS TO IT
\DeclareCaptionStyle{ruled}{labelfont=normalfont,labelsep=colon,strut=off} % DO NOT CHANGE THIS
\usepackage{amsmath, amsthm,amssymb,amsfonts}
\usepackage{algorithm}
\usepackage{algorithmic}
\usepackage{mathtools}
\usepackage{thmtools}
\usepackage{thm-restate}
\usepackage{enumitem}
\usepackage[capitalise]{cleveref}
\usepackage{derivative}
\usepackage{comment}
\usepackage{xcolor}

\usepackage{caption}
\def\mypath{}
\usepackage{newfloat}
\usepackage{listings}
\floatstyle{ruled}
\newfloat{listing}{tb}{lst}{}
\floatname{listing}{Listing}

\def\O{\mathcal{O}}

\def\Pr{\mathbb{P}}
\def\Exp{\mathbb{E}}
\def\Var{\operatorname{Var}}
\def\Nats{\mathbb{N}}

\def\d{\operatorname{d}}
\def\dfp{\fp'(G^S,0)}

\newlist{compactenum}{enumerate}{3} % 3 is max-depth
\setlist[compactenum]{label=\arabic*., nosep,leftmargin=*}
\crefname{compactenumi}{Item}{Items}

\DeclareMathOperator*{\argmax}{arg\,max}

\newtheorem{theorem}{Theorem}



\newcommand{\fit}{f}

\newcommand{\Reals}{\mathbb{R}}

\newcommand{\InNeighbors}{\operatorname{in}}

\newcommand{\InDegree}{\operatorname{d}}
\newcommand{\Configuration}{X}
\newcommand{\RandomConfiguration}{\mathcal{X}}

\newcommand{\BiasedNodes}{S}

\newcommand{\fp}{\operatorname{fp}}
\newcommand{\ft}{\operatorname{T}}

\newcommand{\Trajectory}{\mathcal{T}}

\newcommand{\Paragraph}[1]{{\smallskip\noindent\bf #1}}

\newcommand{\FitAdv}{\delta}

\newcommand{\NP}{\ensuremath{\operatorname{\mathbf{NP}}}}

\newcommand{\Weight}{w}

\newcommand\numberthis{\addtocounter{equation}{1}\tag{\theequation}}
\newcommand{\MC}{\mathcal{M}}

\def\mypath{}
\usepackage{xfrac} % for \sfrac

\title{Maximizing the Probability of Fixation in the Positional Voter Model}

\author {
    Petros Petsinis, %\textsuperscript{\rm 1}
    Andreas Pavlogiannis, %\textsuperscript{\rm 1}
    Panagiotis Karras %\textsuperscript{\rm 1}
}
\affiliations {
    Aarhus University, Aabogade 34, Aarhus, Denmark\\
    \{petsinis,pavlogiannis,panos\}@cs.au.dk
}
\begin{document}
\pagestyle{plain} % use page number for submission
\maketitle

\begin{abstract}
The Voter model is a well-studied stochastic process that models the invasion of a \emph{novel trait}~$A$ (e.g., a new opinion, social meme, genetic mutation, magnetic spin) in a network of individuals (agents, people, genes, particles) carrying an existing \emph{resident trait}~$B$. Individuals change traits by occasionally sampling the trait of a neighbor, while an \emph{invasion bias}~$\delta\geq 0$ expresses the stochastic preference to adopt the novel trait~$A$ over the resident trait~$B$. The strength of an invasion is measured by the probability that eventually the whole population adopts trait~$A$, i.e., the \emph{fixation probability}. In more realistic settings, however, the invasion bias is not ubiquitous, but rather manifested only in parts of the network. For instance, when modeling the spread of a social trait, the invasion bias represents \emph{localized} incentives. In this paper, we generalize the standard biased Voter model to the \emph{positional} Voter model, in which the invasion bias is effectuated only on an arbitrary subset of the network nodes, called \emph{biased nodes}. We study the ensuing optimization problem, which is, given a budget~$k$, to choose~$k$ biased nodes so as to maximize the fixation probability of a randomly occurring invasion. We show that the problem is $\NP$-hard both for finite~$\delta$ and when~$\delta \to \infty$ (strong bias), while the objective function is not submodular in either setting, indicating strong computational hardness. On the other hand, we show that, when~$\delta\to 0$ (weak bias), we can obtain a tight approximation in~$\O(n^{2\omega})$ time, where~$\omega$ is the matrix-multiplication exponent. We complement our theoretical results with an experimental evaluation of some proposed heuristics.
\end{abstract}

\section{Introduction}\label{sec:intro}

Several real-world phenomena involve the emergence and spread of novel traits in populations of interacting individuals of various kinds. For example, such phenomena may concern the propagation of new information in a social network, the sweep of a novel mutation in a genetically homogeneous population, the rise and resolution of spatial conflict, and the diffusion of atomic properties in interacting particle systems. Network science collectively studies such phenomena as \emph{diffusion processes}, whereby the system of interacting individuals is represented as a network and the corresponding process defines the (stochastic, in general) dynamics of local trait spread, from an individual to its neighbors. These dynamics can vary drastically from one application domain to another, and have been studied extensively, e.g., in the cases of influence and information cascades in social networks~\cite{Domingos2001, Kempe2003, Mossel2007, Zhang2020}, epidemic spread~\cite{Kermack1927,Newman2002}, genetic variation in structured populations~\cite{Moran1958, Lieberman2005, Pavlogiannis2018,Tkadlec2021}, and game-theoretic models of antagonistic interaction~\cite{ohtsuki2006simple,IbsenJensen2015}.

One of the most fundamental diffusion processes is the \emph{Voter model}. While introduced to study particle interactions~\cite{Liggett1985} and territorial conflict~\cite{Clifford1973}, its elegance and simplicity have rendered it applicable to a wide variety of other domains (often under other names, such as \emph{imitation updating} or \emph{death-birth dynamics}), including the spread of social traits~\cite{EvenDar2011, Castellano2009, Bhat2019, Durocher2022}, robot coordination and swarm intelligence~\cite{Talamali2021}, and evolutionary dynamics~\cite{Antal2006, ohtsuki2006simple, Hindersin2015,Allen2017,Tkadlec2020,Allen2020}. 

The voter process starts with a homogeneous population of agents (aka \emph{voters}) scattered over an interaction network and carrying a \emph{resident trait}~$B$. A \emph{novel trait}~$A$ invades the population by initially appearing on some random agent(s), and gets diffused in the network by local stochastic updates: each agent occasionally wakes up and updates its own trait by randomly sampling the trait of a neighbor. In general, the trait~$A$ is associated with an \emph{invasion bias} $\delta\geq 0$, which quantifies the stochastic preference of an agent to choose~$A$ over~$B$ while sampling the traits of its neighbors. The process eventually reaches an \emph{absorption state} (aka \emph{consensus state}~\cite{EvenDar2011}), in which~$A$ either \emph{fixates} in the population or \emph{goes extinct}. The key quantity of interest is the probability of \emph{fixation}, which depends on the network structure and the bias~$\delta$.
A key difference with standard cascade models \cite{kempe2005influential} is that the Voter process is \emph{non-progressive}, meaning that the spread of $A$ can both increase and shrink over time.
It can thus express settings such as switching of opinions.

In realistic situations, the invasion bias is not ubiquitous throughout the network, but rather present only in parts. For instance, in the spread of social traits, the invasion bias represents incentives that are naturally local to certain members of the population. Similarly, in the diffusion of magnetic spins, the invasion bias typically stands for the presence of a magnetic field that is felt locally in certain areas. In this paper, we generalize the standard Voter model to faithfully express the locality of such effects, leading to the \emph{positional Voter model}, in which the invasion bias is only present in a (arbitrary) subset of network nodes, called \emph{biased nodes}, and analyze the computational properties thereof. 

The positional Voter model gives rise to the following optimization problem:~given a budget $k$, which $k$ nodes should we bias to maximize the fixation probability of the novel trait? This problem has a natural modeling motivation. For instance, ad placement can tune consumers more receptive to viral (word-of-mouth) marketing \cite{barbieri2013topic,barbieri2014influence,Zhang2020}. Feature selection of a product can increase its appeal to strategic agents and hence maximize its dissemination~\cite{ivanov2017content}. Expert placement can strengthen the robustness of a social network under adversarial attacks~\cite{alon2015robust}. Nutrient placement can be utilized to increase the territorial spread of a certain organism in ecological networks~\cite{Brendborg2022}. In all these settings, optimization can be modeled as selecting a biased set of network nodes/agents, on which one trait has a propagation advantage over the other.

\Paragraph{Our contributions.} 
In this paper we introduce the positional Voter model, and studied the associated optimization problem.
Our main results are as follows.
\begin{compactenum}
\item We show that computing the fixation probability on undirected networks admits a fully-polynomial-time approximation scheme (FPRAS, \cref{thm:approx_fix_prob}).
\item On the negative side, we show that the optimization problem is $\NP$-hard both for finite $\delta$ (general setting)  and as $\delta\to\infty$ (strong bias, \cref{thm:np_hard}), while the objective function is not submodular in either setting (\cref{thm:non_submodularity}).
%, indicating strong computational hardness.
\item We show that, when the network has self-loops (capturing the effect that an individual might choose to remain to their trait), the objective function becomes submodular as~$\delta\to\infty$, hence the optimization problem can be efficiently approximated within a factor~$1-\sfrac1e$ (\cref{thm:strong_bias_approximation}).
\item We show that, as~$\delta\to 0$ (weak bias), we can obtain a tight approximation in~$\O(n^{2\omega})$ time, where~$\omega$ is the matrix-multiplication exponent (\cref{thm:weak_selection}).
\item Lastly, we propose and experimentally evaluate a number of heuristics for maximizing the fixation probability.
\end{compactenum}
Due to limited space, some proofs appear in the~\cref{sec:app}.

\section{Preliminaries}\label{sec:preliminaries}

\Paragraph{Network structures.}
We consider a population of~$n$ agents spread over the nodes of a (generally, directed and weighted) graph~$G = (V, E, \Weight)$, where~$V$ is a set of~$n$ nodes, $E \subseteq V \times V$ is a set of edges capturing interactions between the agents, and~$\Weight \colon E \to \Reals_{>0}$ is a weight function mapping each edge~$(u,v)$ to a real number~$\Weight(u,v)$ denoting the strength of interaction of~$u$ with~$v$. We denote by~$\InNeighbors(u) = \{v \in V \colon (v,u) \in E\}$ the set of incoming neighbors of node~$u$. The \emph{(in)-degree} of node~$u$ is~$\InDegree(u) = \sfrac{1}{|\InNeighbors(u)|}$. For the Voter process to be well-defined on~$G$, $G$ should be \emph{strongly connected}. In some cases, we consider \emph{undirected} and \emph{unweighted} graphs, meaning that (i)~$E$ is symmetric; and (ii)~for all~$(u,v)\in E$, $\Weight(u,v) = 1$; this setting captures networks in which interactions are bidirectional and each node~$v$ is equally influenced by each neighbor. 

\Paragraph{The standard Voter model.}
A \emph{configuration}~$\Configuration\subseteq V$ represents the set of agents that carry the trait~$A$. The trait~$A$ is associated with an \emph{invasion bias}~$\delta \geq 0$ (with $\delta = 0$ denoting no bias, or the \emph{neutral setting}). Given a configuration $\Configuration$, the \emph{influence strength} of node $v$ is defined as
\[
\fit_{\Configuration}(v)=
\begin{cases}
1+\delta, & \text{ if }v\in \Configuration\\
1, & \text{ otherwise.}
\end{cases}
\numberthis\label{eq:standard_bias}
\]
Let~$\RandomConfiguration_t \subseteq V$ be a random variable representing the configuration at time $t$. The Voter (or \emph{death-birth}) process is a discrete-time stochastic process $\{ \RandomConfiguration_t \}$, $t\geq 0$ that models the invasion of a \emph{novel trait}~$A$ on a homogeneous population of agents carrying a \emph{resident trait}~$B$, scattered over the nodes of a graph~$G = (V, E, \Weight)$. Initially, the trait~$A$ appears uniformly at random on one agent, i.e., $\Pr[\RandomConfiguration_0=\{u\}] = \sfrac1n$ for each~$u \in V$. Given the configuration~$\RandomConfiguration_t = \Configuration$ at time~$t$, the next configuration~$\RandomConfiguration_{t+1}$ at time~$t+1$ is determined by a sequence of two stochastic events:
\begin{compactenum}
\item \emph{Death:} an agent~$u$ is chosen to die (i.e., update its trait) with probability $\sfrac1n$.
\item \emph{Birth:} a neighbor~$v \in \InNeighbors(u)$ is chosen with probability
\[
\frac{\fit_{\Configuration}(v)\cdot \Weight(v,u)}{\sum_{x\in \InNeighbors(u)
%\Neighbors(u)
}\fit_{\Configuration}(x)\cdot \Weight(x,u)}
\numberthis\label{eq:birth_prob}
\]
\end{compactenum}
In effect, the set of agents carrying~$A$ may grow or shrink at any given step. In general, we may have~$(u,u)\in E$, expressing the event that agent $u$ stays at its current trait.

\Paragraph{The positional Voter model.}
To capture settings in which the invasion bias is manifested only in parts of the population, we generalize the standard Voter model to the \emph{positional Voter model} by the following two steps:
\begin{compactenum}
\item The network structure comprises a component~$\BiasedNodes\subseteq V$, the subset of agents that are biased to the invasion of~$A$.
\item The influence strength of an agent~$v$ is now \emph{conditioned} on whether the neighbor~$u$ that~$v$ is attempting to influence is biased. Formally, we replace $\fit_{\Configuration}(v)$ with $\fit_{\Configuration}^{\BiasedNodes}(v|u)$ in~\cref{eq:birth_prob}, where:
\[
\fit_{\Configuration}^{\BiasedNodes}(v|u)=
\begin{cases}
1+\delta, & \text{ if }v\in \Configuration \text{ and } u\in \BiasedNodes\\
1, & \text{ otherwise.}
\end{cases}
\numberthis\label{eq:new_bias}
\]
\end{compactenum}
We retrieve the standard Voter model by setting~$\BiasedNodes = V$, i.e., the invasion bias is uniformly present in the population. \cref{fig:pos_voter} illustrates the process.

\begin{figure}[h]
\centering
\includegraphics[width=0.45\textwidth]{\mypath 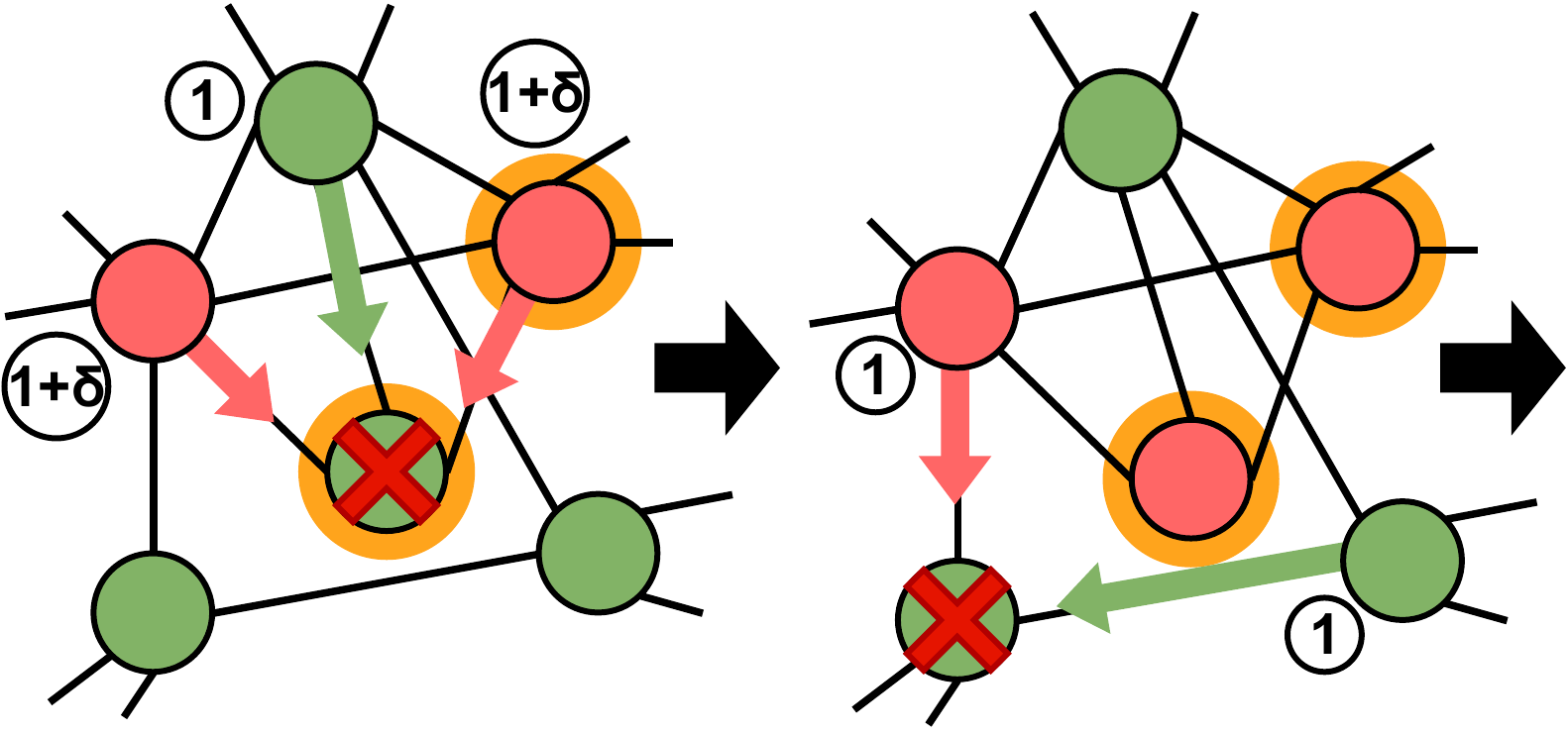}
\caption{Two steps of the positional Voter process. Red (resp., green) nodes carry the mutant trait~$A$ (resp., resident trait~$B$); orange circles mark biased nodes.}\label{fig:pos_voter}
\end{figure}

\Paragraph{Fixation.}
In the long run, the Voter process reaches a homogeneous state almost surely; that is, with probability~$1$, there exists a time~$t$ such that~$\RandomConfiguration_t \in \{\emptyset, V \}$. 
If $\RandomConfiguration_t = V$, the novel trait has \emph{fixated} in~$G$, otherwise it has \emph{gone extinct}. 
Given a configuration $\Configuration$, a biased set~$\BiasedNodes$ and a bias~$\delta$, we denote the probability that~$A$ fixates when starting from nodes~$\Configuration$ as:
\[
 \fp(G^{\BiasedNodes}, \delta, \Configuration)=\Pr[\exists t \ge 0\colon \RandomConfiguration_{t} \!=\! V \mid \RandomConfiguration_{0} \!=\! \Configuration]
\]

As the invasion starts on a random node, the \emph{fixation probability} of~$A$ in~$G$ is the average fixation probability over~$V$:
\[
\fp(G^{\BiasedNodes},\delta)=\frac{1}{n}\sum_{u\in V}\fp(G^{\BiasedNodes},\delta, \{u\})\ .
\]

When $\delta=0$, the set~$\BiasedNodes$ is inconsequential, hence the positional Voter model reduces to the standard model, for which~$\fp(G^{\BiasedNodes}, \delta) = \sfrac1n$~\cite{mcavoy2021}. 
When~$G$ is undirected and~$\BiasedNodes = V$ (i.e., in the standard Voter model), $\fp(G^{\BiasedNodes}, \FitAdv)$ admits a fully polynomial randomized approximation scheme (FPRAS) via Monte-Carlo simulations~\cite{Durocher2022}. 
For arbitrary graphs, however (that may have directed edges or non-uniform edge weights), the complexity of computing~$\fp(G^{\BiasedNodes}, \FitAdv)$ is open. 
We later show that~$\fp(G^{\BiasedNodes},\FitAdv)$ admits an FPRAS for undirected graphs even under the positional Voter model (i.e., for \emph{arbitrary}~$\BiasedNodes$).

\begin{figure}[h]
\centering
\includegraphics[width=0.45\textwidth]{\mypath 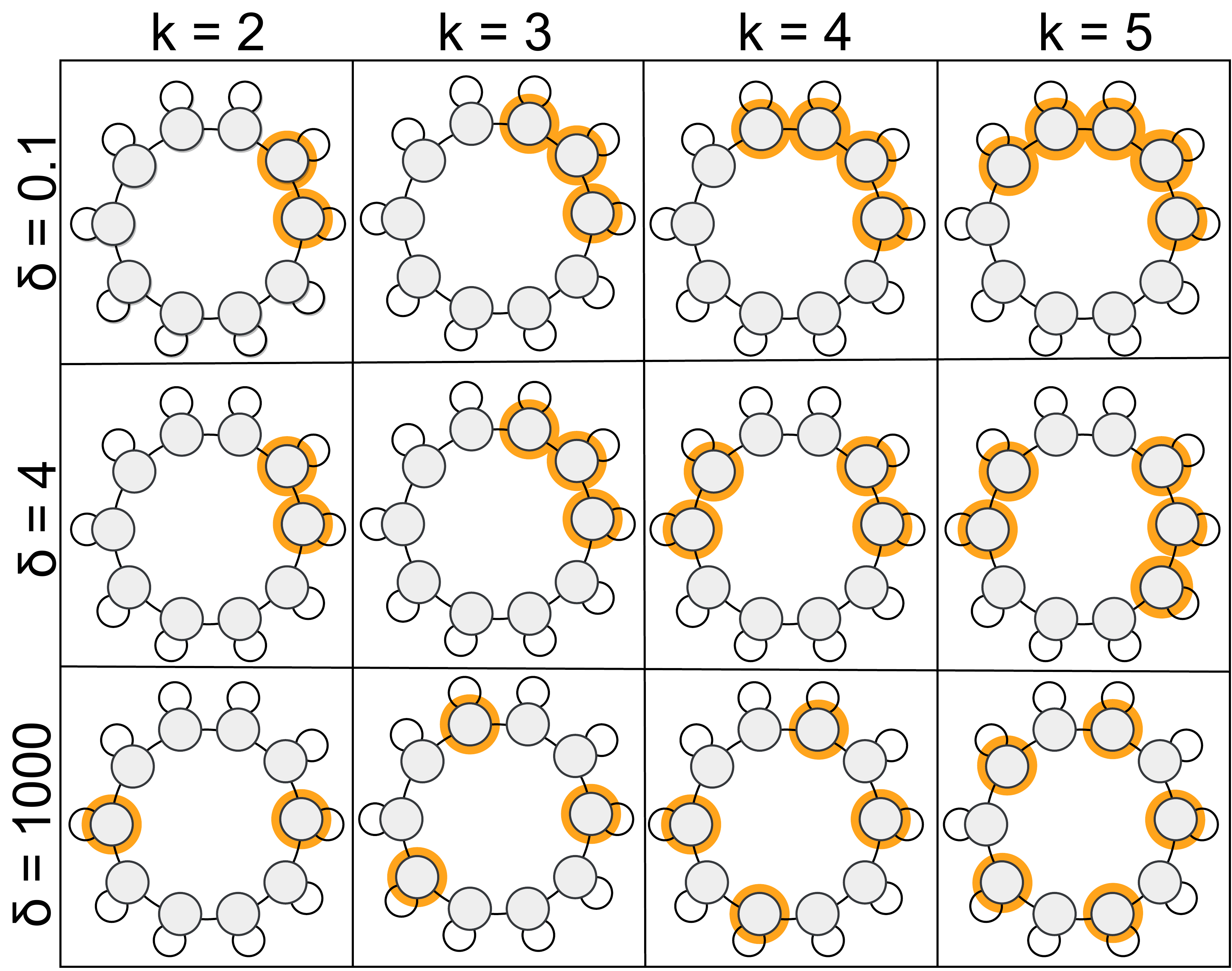}
\caption{Optimal biased node sets~$\BiasedNodes$ (in orange) on a cycle graph for different values of bias~$\delta$ and budget~$k$; for small~$\delta$ (top), the optimal biased nodes are consecutive, yet for large~$\delta$ (bottom), they are spaced apart. For intermediate values (middle), the optimal strategy varies depending on~$k$.}\label{fig:wheel_graph}
\end{figure}

\Paragraph{Optimization.}
In the positional Voter model, $\fp(G^{\BiasedNodes}, \FitAdv)$ depends not only on the network structure~$G$ and invasion bias~$\delta$, but also on the set of biased nodes~$\BiasedNodes$. As~\cref{fig:wheel_graph} illustrates, this dependency can be quite intricate, as the optimal strategy for choosing~$\BiasedNodes$ may vary depending on its size and  the value of~$\delta$. In effect, the positional Voter model naturally gives rise to the following \emph{optimization} problem: Given a budget~$k\in \Nats$, which subset~$\BiasedNodes\subseteq V$ of at most~$k$ nodes should we bias, to maximize the fixation probability? Formally, we seek a set $\BiasedNodes^*\subseteq V$ such that
\[
\BiasedNodes^*=\argmax_{\BiasedNodes\subseteq V, |\BiasedNodes|\leq k}\fp(G^{\BiasedNodes}, \delta)\ .
\numberthis\label{eq:maximization}
\]
As we will show, $\fp(G^{\BiasedNodes}, \delta)$ is monotonic on~$\BiasedNodes$ for all~$\delta$, thus the condition~$|\BiasedNodes|\leq k$ reduces to~$|\BiasedNodes|= k$.
% END of Column 1, Page 3.
We also consider the two extreme cases of $\delta\to \infty$ (\emph{strong bias}) and $\delta\to 0$ (\emph{weak bias}). In the case of strong bias, we define
\[
\fp^{\infty}(G^S)=\lim_{\delta\to\infty} \fp(G^S,\delta),
\]
thus the optimization objective of \cref{eq:maximization} becomes
\[
\BiasedNodes^*=\argmax_{\BiasedNodes\subseteq V, |\BiasedNodes|= k}\fp^{\infty}(G^{\BiasedNodes})\ .
\numberthis\label{eq:maximization_strong}
\]
In the case of weak bias, if $\delta=0$, we have $\fp(G^{\BiasedNodes}, 0) = \sfrac1n$, hence the fixation probability is independent of~$\BiasedNodes$. However, as~$\delta\to 0$ but remains positive, different biased sets will yield different fixation probability. In this case we work with the Taylor expansion of~$\fp(G^{\BiasedNodes},\delta)$ around~$0$, and write
\begin{align}
\fp(G^S,\FitAdv) = \frac{1}{n} + \FitAdv \cdot \dfp  +\O(\FitAdv^2),
\end{align}
where~$\dfp \!=\! \frac{\d}{\d \FitAdv} \Bigr|_{\substack{\FitAdv=0}} \fp(G^S, \FitAdv)$. As~$\delta \to 0$, the lower-order terms $\O(\delta^2)$ approach $0$ faster than the second term. Hence, for sufficiently small positive bias~$\delta$, the optimal placement of bias in the network is the one maximizing the derivative~$\dfp$. In effect, the optimization objective for weak invasion bias becomes
\[
\BiasedNodes^*=\argmax_{\BiasedNodes\subseteq V, |\BiasedNodes|= k}\dfp\ .
\numberthis\label{eq:maximization_weak}
\]
\cref{fig:random_graph} illustrates the behavior of $\fp(G^{\BiasedNodes}, \delta)$ for~$k = 2$ and various values of~$\delta$ on a small graph.

\begin{figure}[H]
\begin{tabular}{@{}c@{}}
\centering
\includegraphics[width=0.1\textwidth]{\mypath 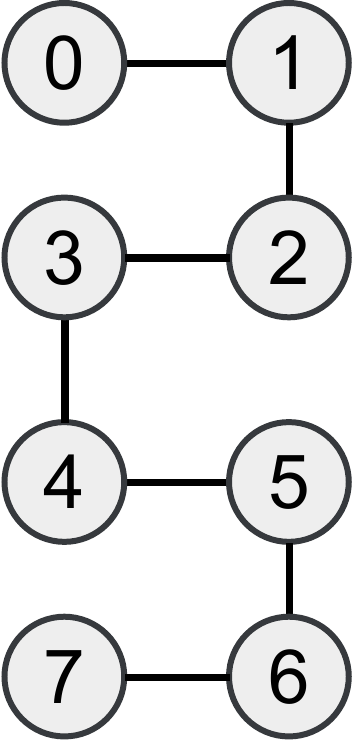}
\end{tabular}
\begin{tabular}{@{}c@{}}
\centering
\includegraphics[width=0.34\textwidth]{\mypath 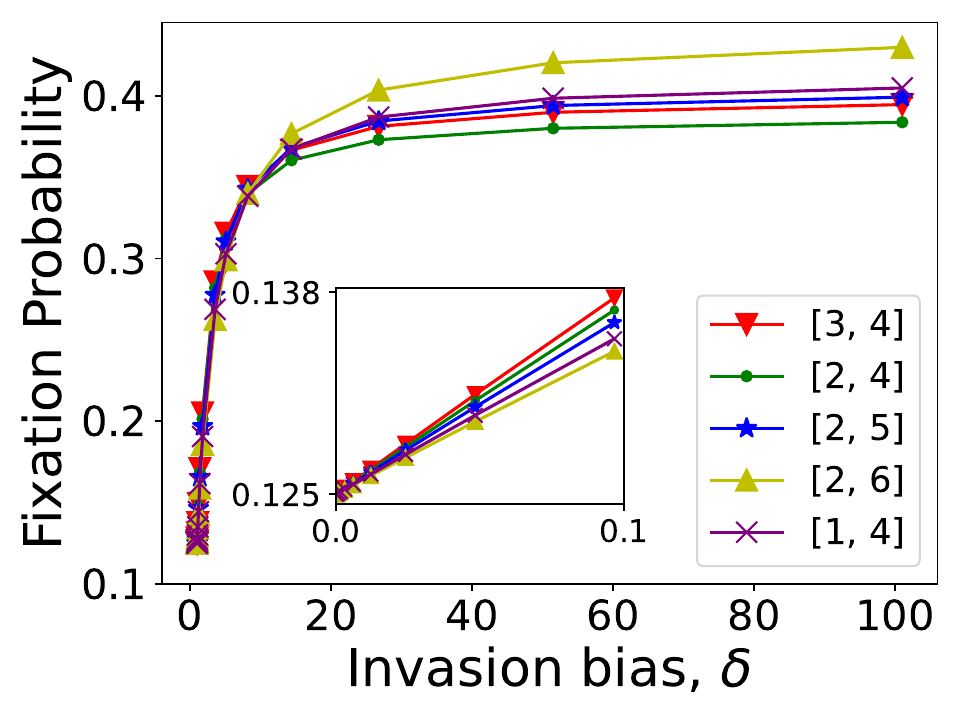} 
\end{tabular}
\caption{Fixation probabilities for different biased sets of size~$k=2$; when~$\delta \in [0, 0.1]$ the fixation probability is roughly linearly dependent on~$\delta$; the set~$\{2,6\}$ is worst for small~$\delta$ but best for large~$\delta$.}\label{fig:random_graph}
\end{figure}

\section{Computing the Fixation Probability}\label{sec:computing_fixation_probability}

In this section we compute the fixation probability~$\fp(G^S, \delta)$ for undirected~$G$. 
In a key step, we show that the expected time~$\ft(G^S,\FitAdv)$ until the positional Voter process reaches a homogeneous state is polynomial in~$n$. 
In particular, we obtain the following lemma that generalizes a similar result for the standard Voter model~\cite{Durocher2022}.

\begin{restatable}{lemma}{thmapproxfixprob}\label{thm:approx_fix_prob}
Given an undirected graph $G$ with $n$ nodes, a set $S\subseteq V$ and some $\FitAdv\geq 1$, we have $\ft(G^S,\FitAdv)\leq n^5$.
\end{restatable}

In effect, we obtain an FPRAS by simulating the process sufficiently many times, and reporting the empirical average number of fixation occurrences as the fixation probability. We thus arrive at the following theorem.

\begin{theorem}\label{thm:approx_fix_prob}
Given a connected undirected graph $G=(V,E)$, a set~$\BiasedNodes\subseteq V$ and a real number~$\FitAdv\geq 0$, the function~$\fp(G^{\BiasedNodes},\FitAdv)$ admits a FPRAS.
\end{theorem}

\section{Hardness of Optimization}

We now turn our attention to the optimization problem for the positional Voter model, and show that it is $\NP$-hard in general. We first examine the process with strong bias, $\delta\to\infty$, running on an undirected, regular graph~$G$ where each node has a self-loop. Our first observation is that, due to self-loops, if the process reaches a configuration~$\Configuration$ with~$\Configuration \cap \BiasedNodes \neq \emptyset$, then fixation is guaranteed.

\begin{restatable}{lemma}{lemstrongbiasfixationcondition}\label{lem:strong_bias_fixation_condition}
Consider an undirected graph~$G$ with self-loops and biased set~$\BiasedNodes$, and let~$\Configuration$ be an arbitrary configuration. If~$\Configuration \cap \BiasedNodes \neq \emptyset$, then~$\fp^{\infty}(G^{\BiasedNodes}, \Configuration) = 1$.
\end{restatable}
\begin{proof}
Consider any node~$u = \Configuration\cap\BiasedNodes$.
For any two configurations~$\Configuration_1$ and~$\Configuration_2$ with~$u \in \Configuration_1$ and~$u \not \in \Configuration_2$, we have
\[
\Pr[\RandomConfiguration_{t+1}=\Configuration_2| \RandomConfiguration_t=\Configuration_1] \leq \frac{1}{n} \frac{1}{\delta}\to 0 \quad \text{as}\quad \delta\to\infty
\]

On the other hand, with probability at least~$(\sfrac1n)^n > 0$, the process reaches fixation within~$n$ steps when starting from any non-empty configuration.
Indeed, while $\emptyset \subset\Configuration\subset V$, with probability at least~$\sfrac1n$, a~$B$-node~$u$ with an~$A$-neighbor is chosen for replacement, and such~$A$-neighbor propagates its trait to~$u$ with probability that approaches~$1$ as~$\delta \to \infty$. 
Thus, the probability that we reach fixation before we reach a configuration~$\Configuration_2$ with~$\Configuration_2 \cap \BiasedNodes = \emptyset$ approaches~$1$ as~$\delta\to\infty$.
\end{proof}

Our second observation relies on~\cref{lem:strong_bias_fixation_condition} to argue that, for an undirected, regular graph~$G$ with self-loops, when the budget~$k$ is sufficiently large,
the optimal choice for the biased set~$\BiasedNodes$ forms a vertex cover of~$G$.

\begin{restatable}{lemma}{lemvertexcover}\label{lem:vertex_cover}
Let~$G = (V, E)$ be an undirected, $d$-regular graph with self-loops, and~$\BiasedNodes \subseteq V$ a biased set. Then~$\fp^{\infty}(G^{\BiasedNodes}) \geq \frac{{|\BiasedNodes|}/{n}+d}{1+d}$
iff $\BiasedNodes$ is a vertex cover of~$G$.
\end{restatable}
    
\begin{proof}
Since~$G$ is connected and has self-loops, we have~$d \geq 2$. 
Due to the uniform initial placement of trait $A$, the probability that it lands on a biased node is~${|\BiasedNodes|}/{n}$. 
Let~$Y$ be the set of nodes in~$V \setminus S$ that have at least one neighbor \emph{not} in~$S$. By~\cref{lem:strong_bias_fixation_condition} we have
\[
\fp^{\infty}(G^{\BiasedNodes}) = \frac{|\BiasedNodes|}{n} + \frac{n-|\BiasedNodes|-|Y|}{n}p + \frac{1}{n} \sum_{u\in Y}\fp^{\infty}(G^{\BiasedNodes}, \{u\}),
\]
where~$p$ is the probability that a node~$u \in V \setminus \BiasedNodes$ with initial trait~$A$ and whose neighbors are all in~$\BiasedNodes$ propagates~$A$ to any of those neighbors before replacing its own trait. 
Let~$p_1$ and~$p_2$ be the probabilities that~$u$ propagates to, and gets replaced by, any of its~$d-1$ neighbors, respectively, and we have $p = \frac{p_1}{p_1+p_2}$.
Moreover,~$p_1=\frac{d-1}{n}1$ (once a neighbor $v$ dies, $u$ propagates its trait to $v$ with probability $1$), and $p_2~=~\frac{1}{n}\frac{d-1}{d}$, leading to $p=\frac{d}{d+1}$.
If~$\BiasedNodes$ is a vertex cover of~$G$, then~$Y = \emptyset$, hence~$\fp^{\infty}(G^{\BiasedNodes}) = \frac{{|\BiasedNodes|}/{n}+d}{1+d}$. 
On the other hand, if~$\BiasedNodes$ is not a vertex cover of~$G$, then~$|Y| \geq 1$. 
Consider a node~$u\in Y$ and let~$v \in \InNeighbors(u)\setminus \BiasedNodes$. 
Observe that~$\fp^{\infty}(G^{\BiasedNodes}, \{u,v\}) < 1$, since with probability at least~$(\frac{1}{n}\frac{1}{d})^2$ the traits on~$u$ and~$v$ get successively replaced by trait $B$. It follows that~$\fp^{\infty}(G^{\BiasedNodes}) < \frac{{|\BiasedNodes|}/{n}+d}{1+d}$, as desired.
\end{proof}
    
Since vertex cover is $\NP$-hard on regular graphs~\cite{Feige2003}, \cref{lem:vertex_cover} implies $\NP$-hardness for maximizing $\fp^{\infty}(G^{\BiasedNodes})$:~given a budget~$k$, $G$ has a vertex cover of size~$k$ iff~$\max_{\BiasedNodes\subseteq V, |\BiasedNodes| = k}\fp^{\infty}(G^{\BiasedNodes}) = \frac{\frac{|\BiasedNodes|}{n}+d}{1+d}$. Moreover, the continuity of~$\fp(G^{\BiasedNodes}, \delta)$ as a function of~$\delta$ implies hardness for finite~$\delta$ too. We thus arrive at the following theorem.

\begin{restatable}{theorem}{thmnphard}\label{thm:np_hard}
The problem of maximizing~$\fp(G^{\BiasedNodes}, \delta)$ and~$\fp^{\infty}(G^{\BiasedNodes})$ in the positional Voter model is $\NP$-hard.
\end{restatable}

\section{Monotonicity and Submodularity}

We now turn our attention to the monotonicity and (conditional) submodularity properties of the fixation probability. Our first lemma formally establishes the intuition that, as we increase the set of biased nodes~$\BiasedNodes$ or the invasion bias~$\delta$, the chances that the novel trait~$A$ fixates do not worsen.

\begin{restatable}{lemma}{lemmonotonicity}\label{lem:monotonicity}
For a graph~$G$, biased sets~$\BiasedNodes_1, \BiasedNodes_2\subseteq V$ with~$\BiasedNodes_1 \subseteq \BiasedNodes_2$ and biases~$\delta_1, \delta_2 \geq 0$ with~$\delta_1 \leq \delta_2$, we have~$\fp(G^{\BiasedNodes_1}(\delta_1)) \leq \fp(G^{\BiasedNodes_2}(\delta_2))$.
\end{restatable} 

\begin{proof}
Consider the two respective Voter processes $\MC_1 = \RandomConfiguration_0^1, \RandomConfiguration_1^1, \dots$ and $\MC_2 = \RandomConfiguration_0^2, \RandomConfiguration_1^2, \dots$. We establish a coupling between~$\MC_1$ and~$\MC_2$ that guarantees~$\RandomConfiguration_t^1 \subseteq \RandomConfiguration_t^2$, for all~$t$, which proves the lemma.

Indeed, assume that the two processes are in configurations~$\RandomConfiguration_t^1 = \Configuration^1$ and $\RandomConfiguration_t^2 = \Configuration_2$, with $\Configuration_1 \subseteq \Configuration_2$. We choose the same node~$u$ to be replaced in the two processes, uniformly at random. Observe that, since~$\BiasedNodes_1 \subseteq \BiasedNodes_2$ and~$\delta_1 \leq \delta_2$, the probability~$p_2$ that~$v$ is replaced by an~$A$-neighbor in~$\MC_2$ is at least as large as the corresponding probability~$p_1$ in~$\MC_1$. Thus, with probability~$p_2 - p_1 \geq 0$, we replace~$u$ in~$\MC_2$ with one of its~$A$-neighbors, leading to a configuration $\Configuration_2' \supseteq \Configuration_2 \supseteq \Configuration_1$. With probability~$p_1$, we replace~$u$ in both~$\MC_1$ and~$\MC_2$ with one of its~$A$-neighbors, leading to configurations~$\Configuration'_1$ and~$\Configuration'_2$ with~$\Configuration'_2\supseteq \Configuration'_1$. Lastly, with probability~$1-p_1$, we replace~$u$ in both~$\MC_1$ and~$\MC_2$ with one of its~$B$-neighbors, again leading to configurations~$\Configuration'_1$ and~$\Configuration'_2$ with~$\Configuration'_2 \supseteq \Configuration'_1$. The desired result follows.
\end{proof}

\Paragraph{Submodularity.}
A real-valued set function~$f$ is called submodular if for any two sets~$S_1, S_2$, we have
\[
f(\BiasedNodes_1) + f(\BiasedNodes_2) \geq f(\BiasedNodes_1\cup \BiasedNodes_2) + f(\BiasedNodes_1\cap \BiasedNodes_2)\ .
\numberthis\label{eq:submodularity_condition}
\]
Submodularity captures a property of diminishing returns, whereby the contribution of an element~$u\in \BiasedNodes_1$ to the value of the function~$f(\BiasedNodes_1)$ decreases as~$\BiasedNodes_1$ increases. The maximization of monotone submodular functions is known to be efficiently approximable, even though it might be intractable to achieve the maximum value~\cite{Nemhauser1978}. In light of our hardness result (\cref{thm:np_hard}) and the monotonicity result (\cref{lem:monotonicity}), it is natural to hope for the submodularity of the fixation probability as a means to at least approximate our maximization problem efficiently. Unfortunately, as the next theorem shows, neither~$\fp(G^{\BiasedNodes}, \delta)$ nor~$\fp^{\infty}(G^{\BiasedNodes})$ exhibit submodularity.

\begin{restatable}{theorem}{thmnonsubmodularity}\label{thm:non_submodularity}
The following assertions hold:
\begin{compactenum}
\item $\fp(G^{\BiasedNodes},\delta)$ is not submodular, and this holds also on graphs with self loops.
\item $\fp^{\infty}(G^{\BiasedNodes})$ is not submodular in general.
\end{compactenum}
\end{restatable} 
\begin{proof}
Our proof is via counterexamples, shown in~\cref{fig:non_submod}. In each case we choose biased sets~$\BiasedNodes_1, \BiasedNodes_2$ that violate the submodularity condition~\cref{eq:submodularity_condition}. We derive the fixation probability exactly using a numerical solver that computes the absorption probabilities in the underlying Markov chain defined by the Voter process on each graph.

\begin{compactenum}
\item Consider the cycle graph of four nodes (\cref{fig:non_submod}, left), with~$\BiasedNodes_1 = \{0,2\}$ and~$\BiasedNodes_2=\{1,3\}$. We calculate
\begin{align*}
&\fp(G^{\BiasedNodes_1}, 0.1) = \fp(G^{\BiasedNodes_2}, 0.1) \leq 0.26194, \text{ while}\\
&\fp(G^{\BiasedNodes_1\cup \BiasedNodes_2}, 0.1) \geq 0.274  \text{ and } \fp(G^{\BiasedNodes_1\cap \BiasedNodes_2}, 0.1) \geq 0.25    
\end{align*}
which violates \cref{eq:submodularity_condition}.
Now consider the same graph with additional self-loops (\cref{fig:non_submod}, middle). We calculate
\begin{align*}
&\fp(G^{\BiasedNodes_1}, 0.1) = \fp(G^{\BiasedNodes_2}, 0.1) \leq 0.2702, \text{ while}\\
&\fp(G^{\BiasedNodes_1\cup \BiasedNodes_2}, 0.1) \geq 0.2909 \text{ and } \fp(G^{\BiasedNodes_1\cap \BiasedNodes_2}, 0.1) \geq 0.25
\end{align*}
which also violates \cref{eq:submodularity_condition}.

\item Consider the wheel graph of 9 nodes (\cref{fig:non_submod}, right), with~$\BiasedNodes_1 = \{1\}$ and~$\BiasedNodes_2 = \{5\}$. We calculate
\begin{align*}
&\fp^{\infty}(G^{\BiasedNodes_1}) = \fp^{\infty}(G^{\BiasedNodes_2}) \leq 0.19, \text{ while}\\
&\fp^{\infty}(G^{\BiasedNodes_1\cup \BiasedNodes_2}) \geq 0.27 \text{ and } \fp^{\infty}(G^{\BiasedNodes_1\cap \BiasedNodes_2}) \geq 0.111
\end{align*}
which violates \cref{eq:submodularity_condition}.
\end{compactenum}
\end{proof}

\begin{figure}[H] 
\centering
\includegraphics[width=0.4\textwidth]{\mypath 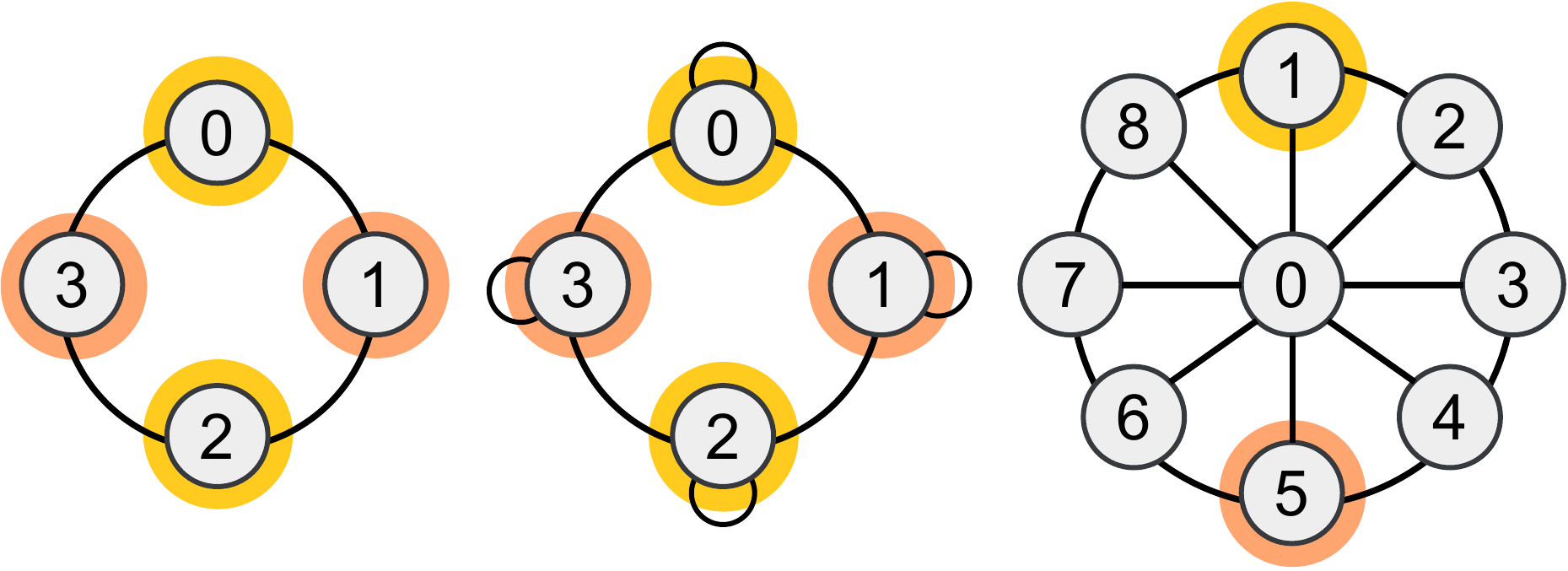}
\caption{Counterexamples to submodularity.}\label{fig:non_submod}
\end{figure}

Interestingly, \cref{thm:non_submodularity} distinguishes the positional Voter and Moran models, as for the latter, $\fp^{\infty}(G^{\BiasedNodes})$ is submodular on arbitrary graphs~\cite{Brendborg2022}. Also, note the asymmetry between $\fp(G^{\BiasedNodes}, \delta)$ and $\fp^{\infty}(G^{\BiasedNodes})$ in \cref{thm:non_submodularity} with regards to graphs having self loops. Indeed, as we show next, $\fp^{\infty}(G^{\BiasedNodes})$ becomes submodular on graphs with self loops (unlike $\fp(G^{\BiasedNodes}, \delta)$).

\begin{restatable}{lemma}{lemsubmodularitystrong}\label{lem:submodularity_strong}
For any undirected graph $G$ with self-loops,
$\fp^{\infty}(G^{S})$ is submodular.
\end{restatable}
\begin{proof}
\emph{A finite trajectory $\Trajectory$} is either a node $\Trajectory=u_0$ or a sequence
$\Trajectory=u_0, (u_1, v_1), (u_2, v_2),\dots, (u_m, v_m), u_{m+1}$,
representing that
\begin{compactenum}
\item trait $A$ starts from $u_0$,
\item for each time $t\in \{1,\dots, m\}$, node $u_t$ adopts the trait of its neighbor $v_t$, and
\item $u_{m+1}$ is chosen for death at time $m+1$.
\end{compactenum}
A prefix $\Trajectory'$ of $\Trajectory$ is $\Trajectory'=u_0$, if $\Trajectory=u_0$,
or $\Trajectory'=u_0, (u_1, v_1), \dots, (u_{t-1}, v_{t-1}), u_t$, for some $t\in[m+1]$.
Thus $\Trajectory'$ follows $\Trajectory$ until the death-step of the $(t+1)$-th event.
We say that $\Trajectory$ is \emph{minimal and fixating} for $\BiasedNodes$ if either $\Trajectory=u_0$ with $u_0\in \BiasedNodes$, or
\begin{compactenum}
\item $u_0\not\in \BiasedNodes$;
\item for each $t\in \{1,\dots,m\}$, we have that either $u_t\not \in \BiasedNodes$, or all neighbors of $u_t$ are $B$-nodes after executing $u_0, (u_1, v_1),\dots (u_t, v_t)$; and
\item $u_{m+1}\in \BiasedNodes$ and $u_{m+1}$ has at least one $A$-neighbor after executing $u_0, (u_1, v_1), (u_2, v_2),\dots, (u_m, v_m)$.
\end{compactenum}
In other words, the last step of~$\Trajectory$ is the first that makes a biased node adopt the novel trait~$A$. Every trajectory that leads to the fixation of~$A$ in~$G^{\BiasedNodes}$ has a \emph{minimal and fixating} prefix. By~\cref{lem:strong_bias_fixation_condition}, the opposite is also true: every minimal and fixating trajectory eventually leads to fixation. Thus, we can compute~$\fp^{\infty}(G^{\BiasedNodes})$ by summing the probabilities of occurrence of each minimal and fixating trajectory~$\Trajectory$.

Moreover, the probability of a minimal and fixating~$\Trajectory$ occurring is independent of~$\BiasedNodes$:~the steps~$u_0$ and~$u_{m+1}$ have probability of~$\sfrac1n$ (since the initial placement of~$A$ is uniform, and each node is chosen for death also uniformly), while each step~$(u_t, v_t)$, for~$t \in \{1, \dots, m\}$ has probability~$\frac{1}{n \cdot \InDegree(u_t)}$. Thus, to arrive at the submodularity of~$\fp^{\infty}(G^{\BiasedNodes})$, it suffices to argue that, for any two biased sets~$\BiasedNodes_1, \BiasedNodes_2$,
\begin{compactenum}
\item\label{item:cup} if $\Trajectory$ is minimal and fixating for $\BiasedNodes_1\cup \BiasedNodes_2$, then it is minimal and fixating for at least one of $\BiasedNodes_1$, $\BiasedNodes_2$, and
\item\label{item:cap} if $\Trajectory$ is minimal and fixating for $\BiasedNodes_1\cap \BiasedNodes_2$, then it has prefixes that are minimal and fixating for both $\BiasedNodes_1$ and $\BiasedNodes_2$.
\end{compactenum}

Indeed, for Item~1, if $\Trajectory=u_0$, then $u_0\in \BiasedNodes_1\cup \BiasedNodes_2$ thus clearly $\Trajectory$ is minimal and fixating for at least one of $\BiasedNodes_1$ and $\BiasedNodes_2$. Similarly, if $\Trajectory=u_0, (u_1, v_1), (u_2, v_2),\dots, (u_m, v_m), u_{m+1}$, then $u_{m+1}$ has an $A$-neighbor in $\BiasedNodes_1\cup \BiasedNodes_2$.
Thus $\Trajectory$ is minimal and fixating for at least one of~$\BiasedNodes_1$ and~$\BiasedNodes_2$, as further, no earlier node~$u_{t}$, for~$t \in [m]$ could have an~$A$-neighbor in~$\BiasedNodes_1 \cup \BiasedNodes_2$.

A similar analysis holds for Item~2. If~$\Trajectory = u_0$, then~$u_0 \in \BiasedNodes_1 \cap \BiasedNodes_2$, and thus~$\Trajectory$ is also minimal and fixating for both~$\BiasedNodes_1$ and~$\BiasedNodes_2$. On the other hand, if~$\Trajectory=u_0, (u_1, v_1), (u_2, v_2),\dots, (u_m, v_m), u_{m+1}$, then~$u_{m+1}$ has an~$A$-neighbor in~$\BiasedNodes_1 \cap \BiasedNodes_2$. Hence, for each set~$Y \in \{\BiasedNodes_1, \BiasedNodes_2\}$, either some earlier node~$u_{t}$ has an~$A$-neighbor in~$Y$ (and thus the prefix~$\Trajectory' = u_0, (u_1, v_1), \dots, (u_{t-1}, v_{t-1}), u_t$ is minimal and fixating for~$Y$), or~$\Trajectory$ is fixating for~$Y$ at~$u_{m+1}$. The desired result follows.
\end{proof}

The monotonicity and submodularity properties of~\cref{lem:monotonicity} and~\cref{lem:submodularity_strong} lead to the following theorem~\cite{Nemhauser1978}.

\begin{restatable}{theorem}{thmstrongbiasapproximation}\label{thm:strong_bias_approximation}
Given an undirected graph~$G$ with self loops and integer~$k$, let~$\BiasedNodes^*$ be the biased set that maximizes~$\fp^{\infty}(G^{\BiasedNodes})$, and~$\BiasedNodes'$ the biased set constructed by a greedy algorithm opting for maximal gains in each step. Then~$\fp^{\infty}(G^{\BiasedNodes'}) \geq (1-\frac{1}{e})\fp^{\infty}(G^{\BiasedNodes^*})$.
\end{restatable}

\section{Optimization for Weak Bias}

In this section we turn our attention to the case of weak bias. Recall that our goal in this setting is to maximize $\dfp$, i.e., the derivative of the fixation probability evaluated at~$\delta=0$.
We show that the problem can be solved efficiently on the class of graphs that have symmetric edge weights (i.e., $\Weight(u,v) = \Weight(v,u)$ for all $u,v\in V$). We arrive at our result by extending the weak-selection method that was developed recently for the basic Voter model in the context of evolutionary dynamics~\cite{Allen2020,allen2021fixation}.

Consider a symmetric graph~$G=(V,E)$ and a biased set~$\BiasedNodes$. 
Given a node~$i\in V$, we write~$\lambda_i=1$ to denote that~$i \in \BiasedNodes$, and~$\lambda_i = 0$ otherwise. 
Given two nodes $i,j\in V$, we let $p_{ij}=\frac{\Weight(i,j)}{\sum_{l\in  V}\Weight(i,l)}$ be the probability that a $1$-step random walk starting in~$i$ ends in~$j$. 
We also let $b_{ij}^{(2)} = \sum_{l\in  V}\lambda_{l}p_{il} \sum_{j\in  V}p_{l j}$ be the probability that a $2$-step random walk that starts in~$i$, goes through a biased node~$l$ and ends in~$j$. 
The following is the key result in this section.

\begin{restatable}{lemma}{lemweakselection}\label{lem:weak_selection}
Consider a symmetric graph $G = (V, E)$, and arbitrary biased set $\BiasedNodes$.
We have
$\dfp = \frac{1}{n}\sum_{i,j\in V}\pi_i b_{ij}^{(2)}\psi_{ij}$, where
$\{\pi_i| i \in V\}$ is the solution to the linear system
\begin{align*}
&\pi_i=\left(1-\sum_{j\in V}p_{ji}\right)\pi_i+\sum_{j\in V}p_{ij}\pi_j,\quad \forall i\in V\\
&\sum_{i\in V}\pi_i=1
\numberthis\label{eq:system_pi}
\end{align*}
and $\{\psi_{ij} | (i, j ) \in E\}$ is the solution to the linear system
\begin{align*}
\psi_{ij}=
\begin{cases}
{\frac{1+\sum_{l\in V}(p_{il}\psi_{lj}+p_{jl}\psi_{il})}{2}} & j\neq i\\
0 & \text{otherwise}
\end{cases}
\numberthis\label{eq:system_psi}
\end{align*}
\end{restatable}

Although \cref{lem:weak_selection} might look somewhat arbitrary at first, the expressions in it have a natural interpretation. We provide this interpretation here, while we refer to the Appendix for the detailed proof.

The quantities~$\pi_i$ express the probability that the novel trait~$A$ fixates when~$\delta=0$ and the invasion starts from node~$i$, and are known to follow \cref{eq:system_pi}~\cite{allen2021fixation}. Since~$\delta=0$, the two traits~$A$ and~$B$ are indistinguishable by the Voter process. Note that, eventually, the whole population will adopt the initial trait of some node~$i$, which leads to~$\sum_{i}\pi_i=1$. The first part of \cref{eq:system_pi} expresses the fact that the current trait of node~$i$ can fixate by either (i)~node~$i$ not adopting a new trait in the current round, and having its trait fixating from the next round on (first term), or (ii)~node~$i$ propagating its trait to some neighbor~$j$ in the current round, and having~$j$'s trait fixate from that round on.

The quantities~$\psi_{ij}$ express the average time throughout the Voter process that nodes~$i$ and~$j$ spend carrying traits~$A$ and~$B$, respectively. First, note that if~$i=j$ (second case of \cref{eq:system_psi}), then clearly~$\psi_{ij}=0$, as~$i$ and~$j$ always have the same trait.
Now, focusing on the first case of \cref{eq:system_psi}, the term~$1$ in the numerator captures the case that the invasion starts at node~$i$, in which case indeed~$i$ and~$j$ carry traits $A$ and $B$ respectively. The second term in the numerator captures the evolution of~$\psi_{ij}$ throughout the process, and is obtained similarly to the~$\pi_i$'s above. Indeed, given a current configuration~$\Configuration$, in the next round~$i$ and~$j$ carry traits~$A$ and~$B$, respectively,
if either
(i)~node~$i$ adopts the trait of some node~$l$, while~$l$ and~$j$ have traits~$A$ and~$B$, respectively (term $p_{il}\psi_{lj}$), or
{(ii)~node~$j$ adopts the trait of some node~$l$, while~$i$ and~$l$ have traits~$A$ and~$B$, respectively (term $p_{lj}\psi_{il}$).}
The denominator $2$ in \cref{eq:system_psi} is a normalizing term.

Now we turn our attention to the expression 
\[
\dfp = \frac{1}{n}\sum_{i,j\in V}\pi_i b_{ij}^{(2)}\psi_{ij} = \frac{1}{n}\sum_{i\in V}\pi_i\sum_{j\in V} b_{ij}^{(2)} \psi_{ij} .
\]
Operationally, this expression can be interpreted as follows:~with probability~$\sfrac1n$, the invasion of~$A$ starts at node~$i$. Then the contribution of the bias~$\delta$ to the fixation of~$i$'s trait~$A$ is multiplicative to the baseline fixation probability~$\pi_i$ of~$i$ under~$\delta = 0$. The multiplicative factor~$\sum_{j\in V}b_{ij}^{(2)}\psi_{ij}$ can be understood by expanding~$b_{ij}^{(2)}$, thereby rewriting as
\[
\sum_{j\in V}b_{ij}^{(2)}\psi_{ij} = \sum_{l\in V} \lambda_{l} p_{il} \sum_{j\in V} p_{lj} \psi_{il}
\]
Node~$i$ benefits by carrying the resident trait~$A$ whenever a neighbor thereof, $l$, is chosen for death. Still, the bias~$\delta$ has an effect if and only if~$l\in \BiasedNodes$ (hence the factor~$\lambda_l$). Moreover, even when~$l \in \BiasedNodes$, the benefit of~$i$ is further proportional to the chance~$\psi_{ij}$ that~$i$ carries trait~$A$ while~$j$ carries trait~$B$ (summed over all neighbors~$j$ of~$l$), since, when~$j$ also carries trait~$A$, it cancels the advantage that~$i$ has due to the bias~$\delta$.

\begin{figure*}[h!]
\begin{minipage}[h!]{0.33\linewidth}
\captionsetup{labelformat=empty,skip=0pt}
k=10\%
\centering
\includegraphics[width=\textwidth]{\mypath 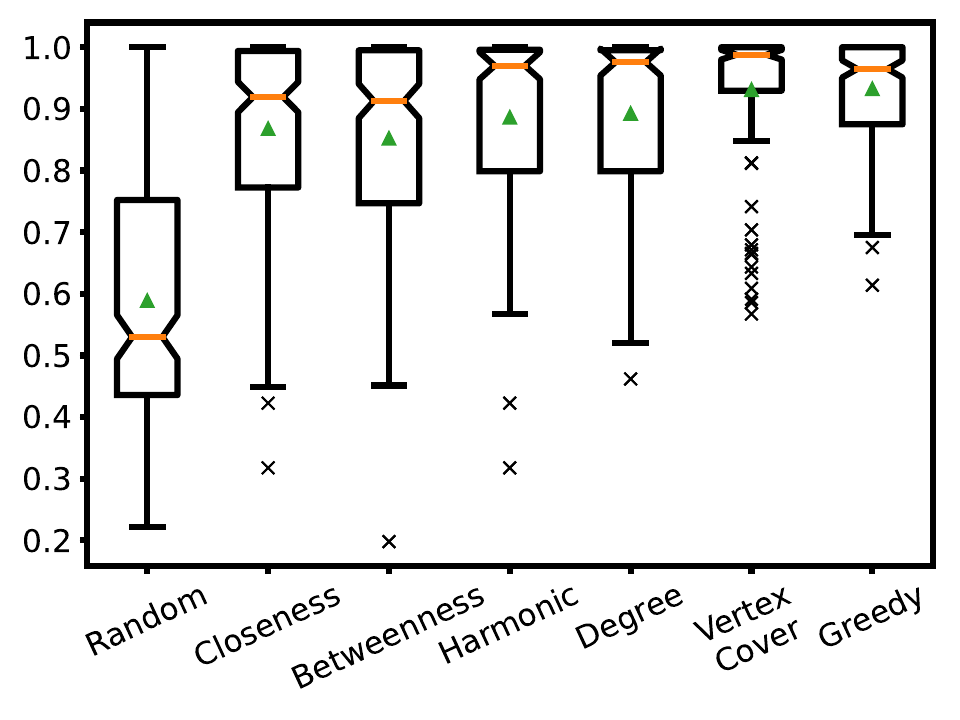}
\end{minipage}
\begin{minipage}[h!]{0.33\linewidth}
\captionsetup{labelformat=empty,skip=0pt}
k=30\%
\centering
\includegraphics[width=\textwidth]{\mypath 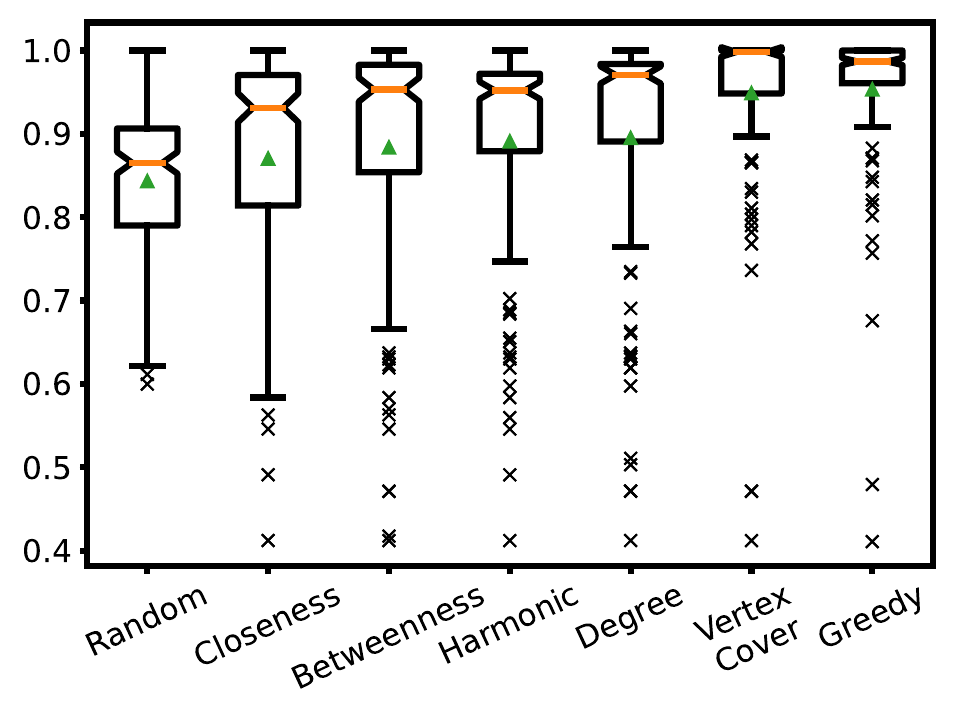}
\end{minipage}
\begin{minipage}[h!]{0.33\linewidth}
\captionsetup{labelformat=empty,skip=0pt}
k=50\%
\centering
\includegraphics[width=\textwidth]{\mypath 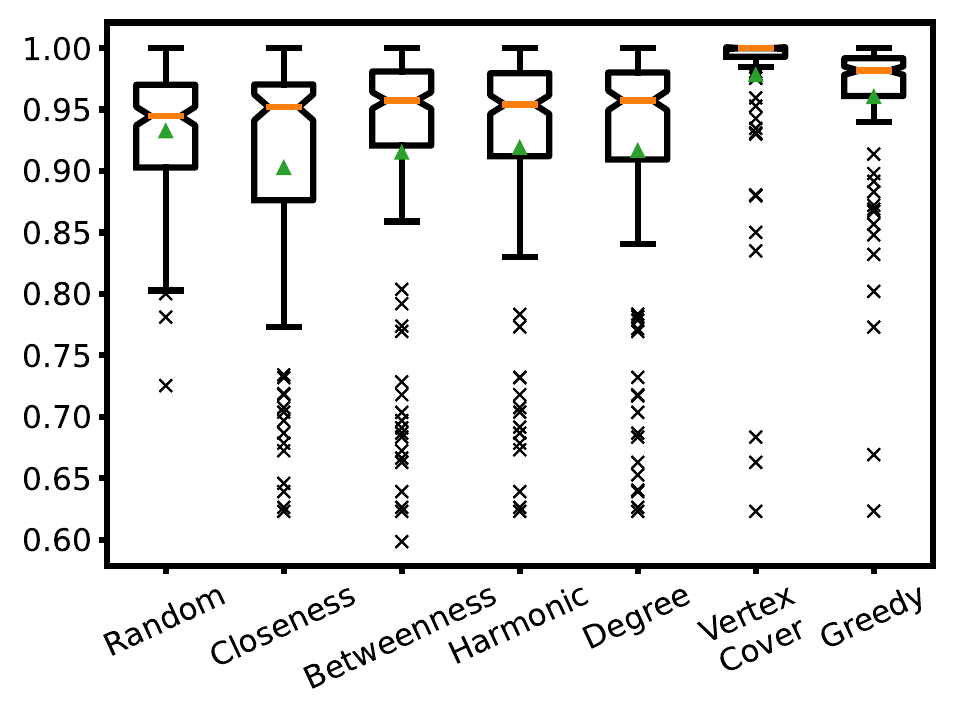}
\end{minipage}
\vspace{-2mm}
\caption{Performance under strong bias ($\delta \to \infty)$}\label{fig:exp_strong}
\vspace{-3mm}
\end{figure*}

\begin{figure*}[h!]
\begin{minipage}[h!]{0.33\linewidth}
\captionsetup{labelformat=empty,skip=0pt}
k=10\%
\centering
\includegraphics[width=\textwidth]{\mypath 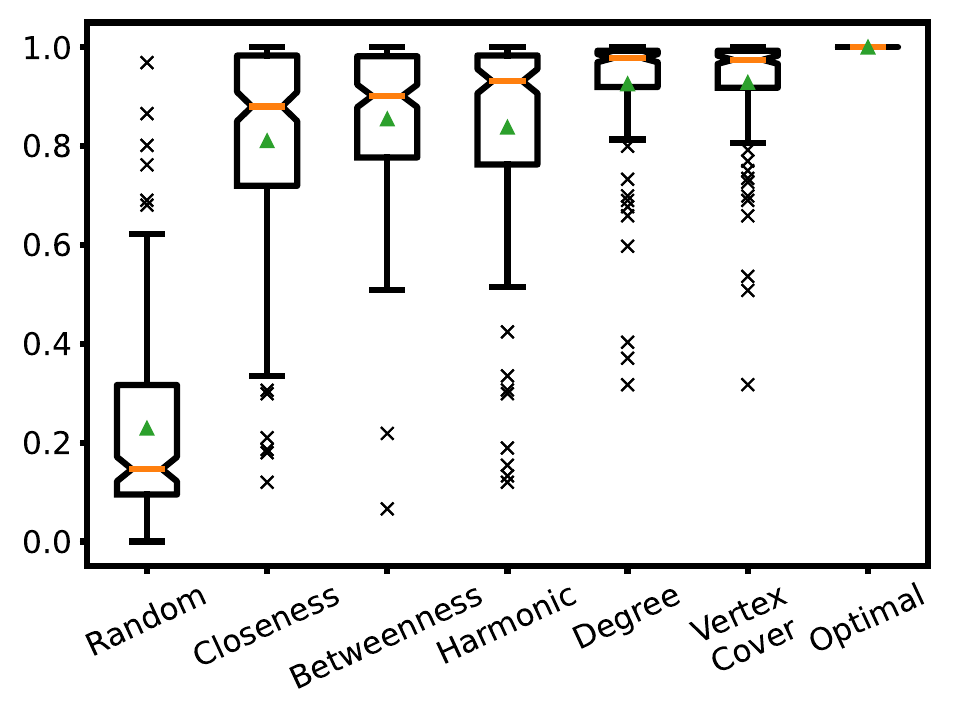}
\end{minipage}
\begin{minipage}[h!]{0.33\linewidth}
\captionsetup{labelformat=empty,skip=0pt}
k=30\%
\centering
\includegraphics[width=\textwidth]{\mypath 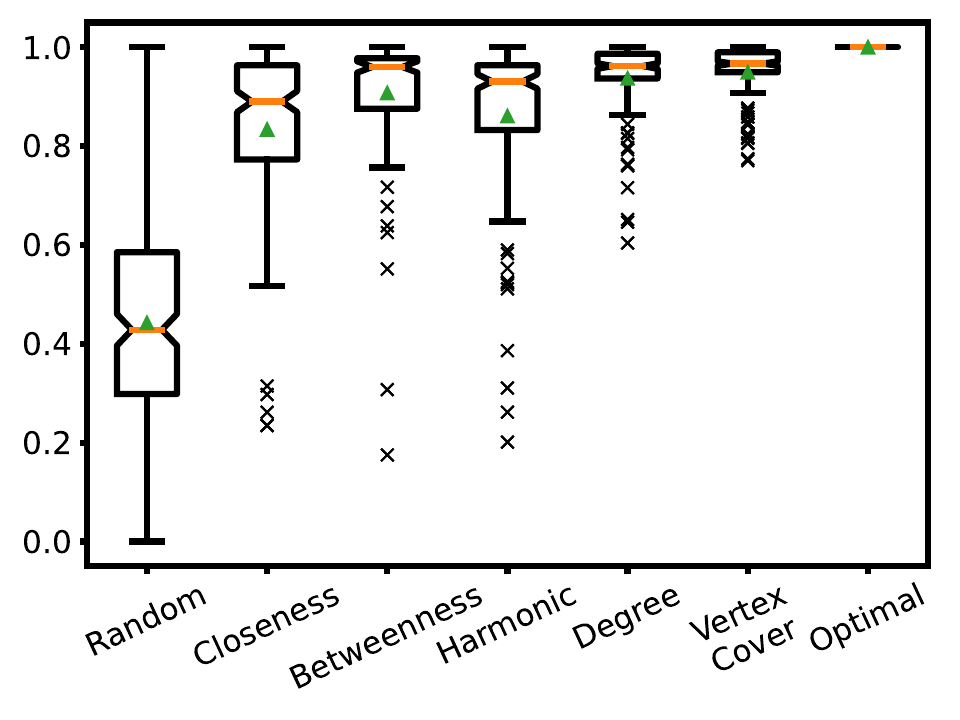}
\end{minipage}
\begin{minipage}[h!]{0.33\linewidth}
\captionsetup{labelformat=empty,skip=0pt}
k=50\%
\centering
\includegraphics[width=\textwidth]{\mypath 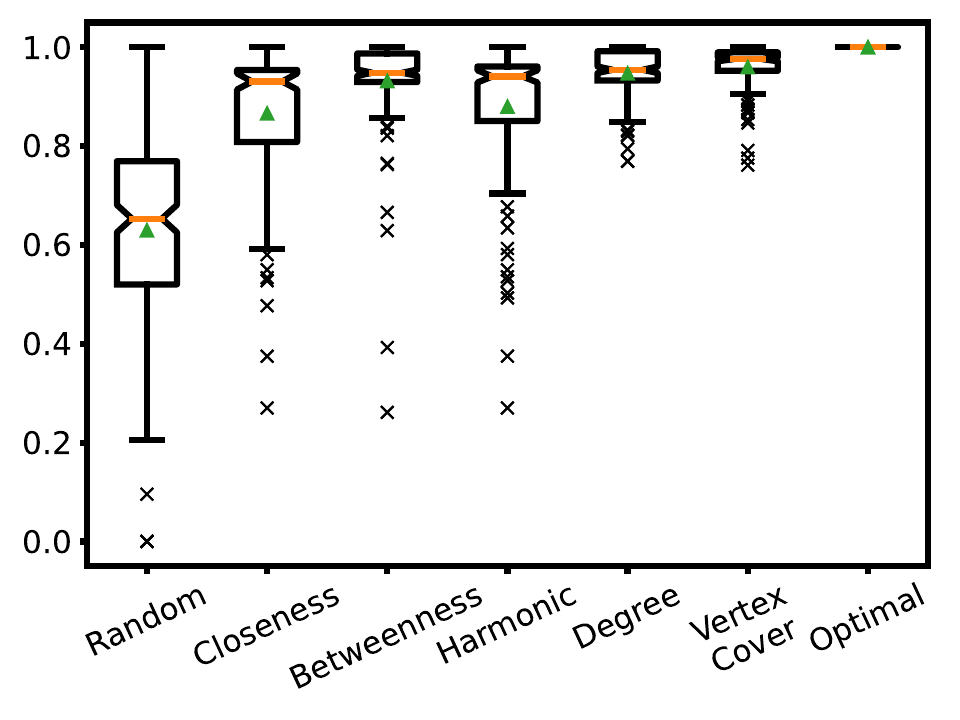}
\end{minipage}
\vspace{-3mm}
\caption{Performance under weak bias ($\delta \to 0)$}\label{fig:exp_weak}
\vspace{-3mm}
\end{figure*}

Having established \cref{lem:weak_selection}, we can now obtain our main result for the case of weak bias.

\begin{restatable}{theorem}{thmweakselection}\label{thm:weak_selection}
Given a symmetric graph~$G$ of~$n$ nodes, the maximization of~$\dfp$ can be done in~$\O(n^{2\omega})$ time, where~$\omega$ is the matrix-multiplication exponent.
\end{restatable}
\begin{proof}
We use standard algorithms to solve the linear systems of \cref{lem:weak_selection} and compute the quantities~$\pi_i$ and~$\psi_{ij}$. In particular, \cref{eq:system_pi} contains~$n$ unknowns and can be solved in~$\O(n^{\omega})$ time, while \cref{eq:system_psi} contains~$n^2$ unknowns and can be solved in~$\O(n^{2\omega})$. Then, by \cref{lem:weak_selection}, we have
\begin{align*}
&\dfp = \frac{1}{n}\sum_{i,j\in V}\pi_i \cdot b_{ij}^{(2)} \cdot\psi_{ij}=\frac{1}{n}\sum_{i\in V} \pi_i \sum_{j\in V}b_{ij}^{(2)} \psi_{il}=\\
&=\frac{1}{n}\sum_{i\in V} \pi_i \sum_{j\in V} \sum_{l\in V} \lambda_j p_{ij} p_{jl} \psi_{il}=\frac{1}{n}\sum_{j\in V} \lambda_{j} h(j),
\numberthis\label{eq:fp_derivative}
\end{align*}
where $h(j)=\sum_{i,l\in V}\pi_i p_{ij} p_{jl} \psi_{il}$. Hence, to maximize~$\dfp$, we compose~$\BiasedNodes$ out of the top-$k$ nodes by~$h(\cdot)$ value.
\end{proof}

\section{Experiments}\label{sec:experiments}

In this section we present experimental results for the proposed algorithms and additional heuristics. 
As datasets, we use~$100$ randomly-selected strongly connected components of real-life social  and community networks taken from~\cite{Netzschleuder}, with varying the number of nodes in the range~$[20,130]$.  
For each graph we use a budget~$k$ corresponding to the $10\%$, $30\%$ and $50\%$ of the graph's size. 
We evaluate 7 algorithms that are often used in a plethora of problems in network analysis.

\begin{compactenum}
\item Random: $k$ nodes uniformly at random.
\item Closeness: Top-$k$ nodes by closeness centrality.
\item Betweenness: Top-$k$ nodes by betweenness centrality.
\item Harmonic: Top-$k$ nodes by harmonic centrality.
\item Degree: Top-$k$ nodes by degree.
\item Vertex Cover: A greedy algorithm that seeks the set $S$ that maximizes the covered edges. The algorithm sequentially selects the node with the maximum marginal gain wrt the number of edges having at least one endpoint in $S$. This heuristic is motivated by \cref{lem:vertex_cover}, which states that vertex cover yields an optimal set under strong bias on regular graphs with self-loops.
\item Greedy: A greedy algorithm that seeks the set $S$ that maximizes the fixation probability. The algorithm sequentially selects the node with the maximal marginal gain wrt fixation probability. To evaluate the gain in each iteration we simulate the process many times to achieve high precision. This algorithm is motivated by \cref{thm:approx_fix_prob}, which indicates that the greedy algorithm provides an~$(1 - \sfrac1e)$-approximation guarantee on graphs with self-loops.%\\
\end{compactenum}

Each algorithm examines different structural characteristics and dynamics. 
Random places the nodes uniformly all over the network. 
Closeness and Harmonic centralities consider the distance of a node to all other nodes. 
On the other hand, Betweenness Centrality indicates the importance of a node serving as a bridge, while Degree only considers the 1-hop neighbors. 
Contrariwise, Vertex Cover considers the collective edge coverage and not an individual score as the preceding strategies do. 
Lastly, the Greedy approach is the only one informed of the positional Voter process.
We compute and report the relative performance of each method, by dividing, for each network, its fixation probability by the maximum fixation probability achieved for that network.

\Paragraph{Strong bias ($\pmb{\delta \to \infty}$)}.
\cref{fig:exp_strong} illustrates the distribution of the normalized results over all graphs. 
We first observe that Greedy and Vertex Cover perform best for all size constraints and have similar behavior. 
This high performance of Greedy and Vertex Cover is expected given the theoretical guarantees from \cref{thm:strong_bias_approximation} and \cref{lem:vertex_cover}, respectively. For small $k$, the problem is quite more challenging, as Random performs poorly. As $k$ increases, a random selection strategy is likely to cover all the graph, and the precise selection is less important (for $k=50\%$ the median is $\geq 95\%$).

\Paragraph{Weak bias ($\pmb{\delta \to 0}$)}. 
By \cref{thm:weak_selection}, to maximize the fixation probability, we need to maximize its derivative~$\dfp$. 
We find the optimal value by solving the linear equation system described in \cref{lem:weak_selection}. 
\cref{fig:exp_weak} presents the normalized results over all graphs. 
We see that Degree and Vertex Cover perform better than centrality algorithms. 
Our intuition is that nodes with many neighbors provide good invasion hubs, as they propagate their trait more frequently. 
Lastly, the Random algorithm under weak selection is ineffective.

\section{Conclusion}

We introduced the \emph{positional Voter model}, which generalizes the standard Voter model to be more faithful to \emph{localized} effects that bias the invasion of novel traits. The new model raises the optimization problem of maximizing the fixation probability by distributing such effects in the network. 
A number of intricate theoretical questions remain open with respect to this optimization problem. 
Can we achieve efficient approximations, despite the lack of submodularity with finite bias~$\delta$? 
Can we obtain such approximations at least for the case of strong bias, under which submodularity holds only for graphs with self-loops? 
Does the tractability under weak bias extend to non-symmetric graphs (i.e., when generally~$\Weight(u,v)\neq\Weight(v,u)$?

\section*{Acknowledgments}
This work was supported in part by DFF (Project 9041-00382B) and Villum Fonden (Project VIL42117).

%\clearpage
%{\small\bibliographystyle{alpha}
\bibliography{\mypath refs}
%\clearpage  %this is not allowd
%\appendix
%\input{\mypath sections/99.appendix}
\appendix
\section{Appendix}\label{sec:app}

\thmapproxfixprob*
\begin{proof}
Consider an arbitrary configuration $\Configuration$ at time $t$, and define the potential function
$\Phi(\Configuration)=\sum_{u\in \Configuration} \InDegree(u)$.
Observe that generally $0\leq \Phi(\Configuration) \leq n^2$.
Consider the random variable $\Delta_t=\Phi(\RandomConfiguration_{t+1})-\Phi(\RandomConfiguration_{t})$,
while $\RandomConfiguration_t$ is realized at an arbitrary configuration $\Configuration$ with $\emptyset\subset \Configuration \subset V$.
For notational simplicity, we denote by $\Pr_{u\to v}$ that node $u$ propagates to node $v$ 

We first argue that $\Exp[\Delta_t]\geq 0$.
Indeed, given any configuration at time $t$, $\RandomConfiguration_t=\Configuration$, 
denote by 
\[
E_{AB}=\{(u,v)\in E\colon u\in \Configuration \text{ and } v\in V\setminus \Configuration \}\ .
\]
Observe that $\Delta_t\neq 0$ only if we have one of the two traits propagated along some edge in $E_{AB}$ in $\RandomConfiguration_{t+1}$.
In particular,
\begin{align*}
\Exp[\Delta_t] =& \sum_{(u,v)\in E_{AB}} \left(\Pr_{u\rightarrow v} \InDegree(v)  -
\Pr_{v\rightarrow u} \InDegree(u)\right)\\
 =& \sum_{(u,v)\in E_{AB}}\frac{1}{n}\left(\frac{\fit_{\Configuration}^{\BiasedNodes}(u|v) \InDegree(v)}{\sum\limits_{x\in \InNeighbors(v)}\fit_{\Configuration}^{\BiasedNodes}(x|v)}-\frac{ \InDegree(u)}{\sum\limits_{x\in \InNeighbors(u)}\fit_{\Configuration}^{\BiasedNodes}(x|u)}\right)
\end{align*}
Note that 
(i)~$\InDegree(u)\leq \sum_{x\in \InNeighbors(u)}\fit_{\Configuration}^{\BiasedNodes}(x|u)$, and
(ii)~$\fit_{\Configuration}^{\BiasedNodes}(u|v) \InDegree(v)\geq \sum_{x\in \InNeighbors(v)}\fit_{\Configuration}^{\BiasedNodes}(x|v)$,
which, by a simple substitution to the above expression for $\Exp[\Delta_t]$, yield $\Exp[\Delta_t]\geq 0$.

Next, we give a lower bound of $1/n$ on the variance of $\Delta(t)$.
Indeed, as long as the population is not homogeneous $\emptyset\subset \Configuration \subset V$, there exists at least one edge $e = (u, v) \in E_{\mathit{AB}}$.
Then with probability at least $\frac{1}{n}\frac{1}{\InDegree(u)}$, we have $\RandomConfiguration_{t+1}=\Configuration\setminus \{u\}$ and thus
$\Pr[\Delta_t\leq -\InDegree(u)]\geq \frac{1}{n \InDegree(u)}$.
Thus, 
\[
\Var(\Delta_t)\geq \Pr_{v\to u}(-\InDegree(u)-\Exp[\Delta_t])^2 \geq \frac{\InDegree(u)^2}{n \InDegree(u)} \geq \frac{1}{n}
\]
as $\InDegree(u)\geq 1$.

The potential function $\Phi$ thus gives rise to a submartingale with bound $B=n^2$.
Using standard martingale machinery of drift analysis~\cite{Kotz2019}
the re-scaled function $\Phi(\Phi-2B)+B^2$ satisfies the conditions
of the upper additive drift theorem with initial value at most
$B^2$ and step-wise drift at least $K=1/n$. 
The expected time until termination is thus at most:
\[
\ft(G^S,\FitAdv) \leq \frac{B^2}{K} = \frac{n^4}{1/n} = n^5
\]
The desired result follows.
\end{proof}

\lemweakselection*
\begin{proof}
Given a configuration $\Configuration\subseteq V$, for each node $i\in V$, we let $x_i=1$ if $i\in \Configuration$, and $x_i=0$ otherwise.
We write $\Phi(\Configuration)$ for the probability that $A$ fixates in $G$ starting from configuration $\Configuration$, when $\delta=0$.
In particular, $\Phi(\Configuration)=\sum_{i \in  V} \pi_i x_i$. 
We study how $\Phi(\Configuration)$ changes in one step of the Voter process when $\delta>0$.
We have
\begin{align*}
\Delta\Phi(X,\delta)&=\Exp[\Phi(\RandomConfiguration_{t+1})-\Phi(\RandomConfiguration_{t})|\RandomConfiguration_t=\Configuration]\\
&=\sum_{i,j\in V}e_{ij}(\Configuration,\delta)(x_i-x_j)\pi_j    
\end{align*}
where
\[
e_{ij}(\Configuration,\delta)=\frac{1}{n}\frac{\Weight(i,j)(1+\delta x_i\lambda_j)}{\sum\limits_{l\in  V}\Weight(l,j)(1+\delta x_l\lambda_j)}
\]
is the probability that node $j$ adopts the trait of node $i$ in the current round.
Expanding $\Delta\Phi(X,\delta)$, we obtain: 
\begin{align*}
&\Delta\Phi(\Configuration,\delta)=\sum_{i,j\in V}e_{ij}(\Configuration,\delta)(x_i-x_j)\pi_j\\
&=\sum_{i,j\in V}\left(\frac{1}{n}\frac{\Weight(i,j)(1+\delta x_i\lambda_j)}{\sum\limits_{l\in  V}\Weight(l,j)(1+\delta x_l\lambda_j)}\right)(x_i-x_j)\pi_j\\
&=\frac{1}{n}\sum_{j\in V}\pi_j\left(-x_j+\sum_{i\in V}x_i\frac{\Weight(i,j)(1+\delta x_i\lambda_j)}{\sum\limits_{l\in  V}\Weight(l,j)(1+\delta x_l\lambda_j)}\right)\\
&=\frac{1}{n}\sum_{i\in V}x_i\left(-\pi_i+\sum_{j\in V}\pi_j\frac{\Weight(i,j)(1+\delta x_i\lambda_j)}{\sum\limits_{l\in  V}\Weight(l,j)(1+\delta x_l\lambda_j)}\right)
\end{align*}

Using the symmetry of the weights $w_{ij}=w_{ji}$, we continue by differentiating $\Delta\Phi(X,\delta)$ with respect to $\delta$ as
$\Delta\Phi'(\Configuration,\delta)=\frac{1}{n}\sum_{i,j\in V}x_i\pi_j(C_1-C_2)$, where $C_1$ and $C_2$ are
\begin{align*}
C_1 &= \frac{\lambda_jx_i\Weight(i,j)}{\sum\limits_{l\in  V}\Weight(l,j)(1+\delta\lambda_j x_l)}\frac{\sum\limits_{l\in  V}\Weight(l,j)(1+\delta\lambda_j x_l)}{\sum\limits_{l\in  V}\Weight(l,j)(1+\delta\lambda_j x_l)}\\
&=\frac{\lambda_jx_i\Weight(j,i)}{\sum\limits_{l\in  V}\Weight(j,l)(1+\delta\lambda_j x_l)}\\
C_2&=\frac{(1+\delta\lambda_j x_i)\Weight(i,j)}{\sum\limits_{l\in  V}\Weight(l,j)(1+\delta\lambda_j x_l)}\frac{\sum\limits_{l\in  V}\Weight(l,j) x_l\lambda_j}{\sum\limits_{l\in  V}\Weight(l,j)(1+\delta\lambda_j x_l)}\\
&=\frac{(1+\delta\lambda_j x_i)\Weight(j,i)}{\sum\limits_{l\in  V}\Weight(j,l)(1+\delta\lambda_j x_l)}\frac{\sum\limits_{l\in  V}\Weight(j,l) x_l\lambda_j}{\sum\limits_{l\in  V}\Weight(j,l)(1+\delta\lambda_j x_l)}
\end{align*}
Setting $\delta=0$, and using the identity $x^2_i=x_i$ and the reversibility property $\pi_ip_{ij}=\pi_jp_{ji}$~\cite{mcavoy2021}, the quantities $C_1$ and $C_2$ become
\begin{align*}
&C_1 = \frac{\lambda_jx_i\Weight(j,i)}{\sum\limits_{l\in V}\Weight(j,l)}=\lambda_jx_ip_{ji}\\
&C_2 = \frac{\Weight(j,i)}{\sum\limits_{l\in  V}\Weight(j,l)}\frac{\sum\limits_{l\in  V}\Weight(j,l) x_l\lambda_j}{\sum\limits_{l\in  V}\Weight(j,l)}=\lambda_jp_{ji}\sum_{l\in V}p_{jl}x_l\\
\end{align*}
Thus the derivative value $\Delta\Phi'(X,0)$ is given by:
\begin{align*}
\Delta\Phi'(X,0)&=\frac{1}{n}\sum_{i,j\in V}x_i\pi_j(\lambda_jx_ip_{ji}-\lambda_jp_{ji}\sum_{l\in V}p_{jl}x_l)\\
&=\frac{1}{n}\sum_{i\in V}x_i\sum_{j\in V}\Big(\pi_jp_{ji}\lambda_j-\pi_jp_{ji}\lambda_j\sum\limits_{l\in  V}p_{jl} x_l\Big)\\
&=\frac{1}{n}\sum_{i\in V}x_i\sum_{j\in V}\Big(\pi_jp_{ji}\lambda_j\Big(1-\sum\limits_{l\in  V}p_{jl} x_l\Big)\Big)\\
&=\frac{1}{n}\sum_{i\in V}x_i\sum_{j\in V}\Big(\pi_jp_{ji}\lambda_j\Big(\sum\limits_{l\in  V}p_{jl}-\sum\limits_{l\in  V}p_{jl} x_l\Big)\Big)\\
&=\frac{1}{n}\sum_{i\in V}x_i\sum_{j\in V}\Big(\pi_ip_{ij}\lambda_j\Big(\sum\limits_{l\in  V}p_{jl}(1-x_l)\Big)\Big)\\
&=\frac{1}{n}\sum_{i\in V}x_i\pi_i\sum_{j,l\in V}\lambda_jp_{ij}p_{jl}(1-x_l)\\
&=\frac{1}{n}\sum_{i\in V}x_i\pi_i\sum_{j\in V}b_{ij}^{(2)}(1-x_j)\\
&=\frac{1}{n}\sum_{i,j\in  V} \pi_i x_i (1-x_j) b^{(2)}_{ij}
\end{align*}

For two nodes $i,j \in V$, we define a stochastic function
\[
Y_{ij}(t)=
\begin{cases}
1, &\text{ if } i\in \RandomConfiguration_t \text{ and } j\not \in \RandomConfiguration_t,\\
0, &\text{ otherwise}
\end{cases}
\]
We also use $\psi_{ij}=\sum_{t=0}^{\infty}\Pr[Y_{ij}(t)]$ to express the total expected time that $i$ and $j$ hold traits $A$ and $B$, respectively,
Using~\cite[Theorem~1]{mcavoy2021} we have
\begin{align*}
\dfp&=\sum_{t=0}^{\infty}\sum_{\Configuration\subseteq V}\Pr[\RandomConfiguration_t=\Configuration]\Delta\Phi'(X,0)\\
&=\frac{1}{n}\sum_{i,j\in V}\pi_i  b_{ij}^{(2)} \psi_{ij}
\end{align*}  

The above equation expresses the crucial property that we can compute $\dfp$ by computing the quantities $\psi_{ij}$.
Thus, it remains to argue that these quantities follow the linear system of \cref{eq:system_psi}.
Indeed, if $i=j$ then clearly $\psi_{ij}=0$. 
For $i\neq j$ the expected time $\psi_{ij}$ follows the recurrence
\begin{align*}
\psi_{ij}=\sum_{t=0}^{\infty}\Pr[Y_{ij}(t)]=\Pr[Y_{ij}(0)]+\sum_{t=0}^{\infty}\Pr[Y_{ij}(t+1)]
\end{align*}

Let 
$q_{ij}=e_{ij}(\Configuration,0)=\frac{1}{n} p_{ij}$
be the probability that $j$ adopts the trait of $i$ in one step under the neutral setting $\delta=0.$
Then, we can rewrite the second term using the recursive formula
\begin{align*}
\Pr[Y_{ij}(t+1)]&=\sum_{l\in  V}q_{li}\Pr[Y_{lj}(t)]+\sum_{l\in  V}q_{lj}\Pr[Y_{il}(t)]\\
&+\left(1-\sum_{l\in V}(q_{li}+q_{lj})\right) \Pr[Y_{ij}(t)]
\end{align*}
Summing over all times $t$, we have
\begin{align*}
\sum_{t=0}^{\infty}\Pr[Y_{ij}(t+1)]&=\sum_{t=0}^{\infty}\Bigg(\sum_{l\in  V}q_{li}\Pr[Y_{lj}(t)]+\sum_{l\in  V}q_{lj}\Pr[Y_{il}(t)]\\
&+\left(1-\sum_{l\in V}(q_{li}+q_{lj})\right) \Pr[Y_{ij}(t)]\Bigg)
\end{align*}
Distributing the sum over $t$, we have
\begin{align*}
\sum_{t=0}^{\infty}\Pr[Y_{ij}(t+1)]&=\sum_{l\in  V}q_{li}\sum_{t=0}^{\infty}\Pr[Y_{lj}(t)]+\sum_{l\in  V}q_{lj}\sum_{t=0}^{\infty}\Pr[Y_{il}(t)]\\
&+\left(1-\sum_{l\in V}(q_{li}+q_{lj})\right) \sum_{t=0}^{\infty}\Pr[Y_{ij}(t)]
\end{align*}
Note that $\sum_{t=0}^{\infty}\Pr[Y_{ij}(t+1)]= \sum_{t=0}^{\infty}\Pr[Y_{ij}(t)] - \Pr[Y_{ij}(0)]$.
Due to the random initialization of the novel trait $A$, we have $\Pr[Y_{ij}(0)]=1/n$.
Moreover, now each sum over $t$ goes from $0$ to $\infty$, and thus we can replace these sums with the corresponding quantities $\psi$.
We thus have
\begin{align*}
\psi_{ij}&=\frac{1}{n}+\sum_{l\in  V}q_{li}\psi_{lj}+
\sum_{l\in  V}q_{lj}\psi_{il}\\
&+\left(1-\sum_{l\in V}(q_{li}+q_{lj})\right)\psi_{ij}, 
\end{align*}
and by rearranging, $\psi_{ij}$ satisfy the recurrence
\begin{align*}
     \psi_{ij}=
    \begin{cases}
    {\frac{1/n+\sum_{l\in V}(q_{li}\psi_{lj}+q_{lj}\psi_{il})}{\sum_{l\in V}(q_{li}+q_{lj})}} & j\neq i\\
    0 & \text{otherwise}
\end{cases}
\end{align*}
or equivalently, the linear system of \cref{eq:system_psi}.
\end{proof}

\end{document}